\documentclass[runningheads]{llncs}
\usepackage{graphicx}
\usepackage{amsfonts}
\usepackage{amssymb}
\usepackage{amstext}
\usepackage{amsmath}
\usepackage{epsfig}
\usepackage{xspace}
\usepackage{graphicx}
\usepackage{graphics}
\usepackage{colordvi}
\usepackage{color}
\usepackage{hyperref}
\usepackage{paralist}
\usepackage{bm}
\usepackage{mathpazo}
\usepackage{algorithm}
\usepackage{algpseudocode}
\usepackage[normalem]{ulem}
\usepackage[dvipsnames]{xcolor}
\usepackage{cancel}
\usepackage[nice]{nicefrac}

\newenvironment{proofof}[1]{\noindent{\bf Proof of #1.}}%
        {\hspace*{\fill}$\Box$\par\vspace{4mm}}

\newcommand{\bigOmega}[1]{\ensuremath{\Omega\left(#1\right)}}

\newcommand{\saplg}{\mbox{\sf SA}$^+_\ell(G)$}
\newcommand{\saplgn}{\mbox{\sf SA}$^+_\ell(G_n)$}
\newcommand{\yh}{\hat{y}}
\newcommand{\qstab}{{\sf QSTAB}\xspace}

\newcommand{\be}{\begin{enumerate}}
\newcommand{\ee}{\end{enumerate}}
\newcommand{\bd}{\begin{description}}
\newcommand{\ed}{\end{description}}
\newcommand{\bi}{\begin{itemize}}
\newcommand{\ei}{\end{itemize}}

\def\square{\vbox{\hrule height.2pt\hbox{\vrule width.2pt height5pt \kern5pt
\vrule width.2pt} \hrule height.2pt}}

\newenvironment{prog}[1]{
\begin{minipage}{5.8 in}
{\sc\bf #1}
\begin{enumerate}}
{
\end{enumerate}
\end{minipage}
}

\newcommand{\cC}{\ensuremath{\mathcal{C}}}

\renewcommand{\phi}{\varphi}
\newcommand{\eps}{\epsilon}

\newcommand{\ksp}{$k$\mbox{\sf -SetPacking}\xspace}

\newcommand{\ignore}[1]{}

\begin{document}
\title{Independent set in $k$-Claw-Free Graphs: Conditional $\chi$-boundedness and the Power of LP/SDP Relaxations\thanks{This research was partially done during the 
trimester on Discrete Optimization at Hausdorff Research Institute for
Mathematics (HIM) in Bonn, Germany. The research has received funding from the European Research Council (ERC) under the European Union’s Horizon 2020 research and innovation programme (grant agreement No 759557) and from  Academy of Finland (grant number  310415). Kamyar Khodamoradi was supported by Deutsche Forschungsgemeinschaft (project number 399223600).}}
\titlerunning{Independent set in $k$-Claw-Free Graphs}
%
\author{Parinya Chalermsook\inst{1} \and
Ameet Gadekar\inst{1} \and
Kamyar Khodamoradi \inst{2} \and
Joachim Spoerhase\inst{3}}
\authorrunning{P. Chalermsook et al.}
%
\institute{Aalto University, Espoo, Finland\\ \email{\{parinya.chalermsook,ameet.gadekar\}@aalto.fi}  \and University of Regina, Canada\\ \email{kamyar.khodamoradi@uregina.ca}
\and  University of Sheffield, UK\\ \email{j.spoerhase@sheffield.ac.uk}}
%
\maketitle              
\begin{abstract}
This paper studies $k$-claw-free graphs, exploring the connection between an extremal combinatorics question and the power of a convex program in approximating the maximum-weight independent set in this graph class. 
For the extremal question, we consider the notion, that we call \textit{conditional $\chi$-boundedness} of a graph: Given a graph $G$ that is assumed to contain an independent set of a certain (constant) size, we are interested in upper bounding the chromatic number in terms of the clique number of $G$. 
This question, besides being interesting on its own, has algorithmic implications (which have been relatively neglected in the literature) on the performance of SDP relaxations in estimating the value of maximum-weight independent set. 

For $k=3$, Chudnovsky and Seymour (JCTB 2010)  prove that any $3$-claw-free graph $G$ with an independent set of size three must satisfy $\chi(G) \leq 2 \omega(G)$. Their result  implies a factor $2$-estimation algorithm for the maximum weight independent set via an SDP relaxation (providing the first non-trivial result for maximum-weight independent set in such graphs via a convex relaxation). An obvious open question is whether a similar conditional $\chi$-boundedness phenomenon holds for any $k$-claw-free graph.  
Our main result answers this question negatively.
We further present some evidence that our construction could be useful in  studying more broadly the power of convex relaxations in the context of approximating maximum weight independent set in $k$-claw free graphs.  In particular, we prove a lower bound  on families of convex programs that are  stronger than  known convex relaxations used algorithmically in this context. 
\end{abstract}
\section{Introduction}

For $k \geq 3$, a graph is said to be $k$-claw-free if the neighborhood of each vertex does not contain an independent set of size $k$. 
This paper focuses on an extremal question in $k$-claw-free graphs and its connection to the power of convex programs in estimating the maximum weight independent set (MWIS) in such graphs. 
The study of such relation originated already around 50 years ago when Lov\'{a}sz defined the notion of perfect graphs based on graph extremal properties and showed connections to (exact) semi-definite programming {formulations} for optimization problems~\cite{grotschel1984polynomial,lovasz1972characterization}. 
Such connections are known to be ``approximation preserving'' so they  imply a connection between standard Ramsey-type {theorems} (and $\chi$-boundedness) in approximating the cardinality (resp. the weight) of the maximum independent set. 
Several approximation algorithms in geometric intersection graphs have been successfully derived in this framework~\cite{lewin2002routing,chalermsook2009maximum,chalermsook2021coloring,chan2009approximation}. 

Most prior works that extend  perfect graphs rely on the notion of the clique constrained stable set polytope (QSTAB)---a convex relaxation that can be optimized via semi-definite programs.\footnote{Recall that the polytope is defined as ${\sf QSTAB}(G)= \{x \in [0,1]^{|V(G)|}: \sum_{i \in Q} x_i \leq 1 (\forall \mbox{clique $Q$})\}$. Optimizing this itself is NP-hard, but we can optimize an SDP whose solution is feasible for QSTAB.}  The power of QSTAB is captured precisely by standard extremal bounds. For example,
 the $\chi$-boundedness  $\chi(G) \leq \gamma \omega(G)$ in a ``natural'' graph class is (roughly) equivalent to QSTAB providing $\gamma$-estimate on the weight of maximum independent set in the same graph class~\cite{chalermsook2016note}. 
 Despite successful cases of the extremal approach, QSTAB fails unexpectedly  in graph classes such as $k$-claw-free graphs: For any $k\geq 3$, a simple greedy algorithm immediately gives a factor $(k-1)$ approximation for MWIS, while QSTAB (and other known convex relaxations) is unable to give $f(k)$ approximation for any function $f$ (see Appendix~\ref{sec: qstab bad}).

This work is an attempt to better understand the power of convex relaxations for approximating MWIS in $k$-claw-free graphs. MWIS on $k$-claw-free graphs contains many well-known (open) problems as special cases, such as set packing~\cite{cygan2013improved,chandra2001greedy,hazan2006complexity,hochbaum1983efficient} and independent set in sparse graphs~\cite{bansal2015lovasz,austrin2011inapproximability}, for which QSTAB has been shown to perform relatively well in terms of approximating the problems. Somewhat surprisingly, the study of convex relaxations for MWIS on $k$-claw-free graphs has been  absent from the literature.

This paper is inspired by the following theorem of Chudnovsky and Seymour.

\begin{theorem}[Chudnovsky--Seymour~\cite{CHUDNOVSKY2010560}]
\label{thm:CS}
For every connected claw-free $(k=3)$ graph $G$ with $\alpha(G) \ge 3$, it holds that $\chi(G) \le 2 \omega(G)$.
\end{theorem}

Remark that, the condition $\alpha(G) \geq 3$ is necessary, for otherwise, we can have $\chi(G) =\widetilde{\Omega}(\omega(G)^2)$ for claw-free graphs, contradicting the above theorem.\footnote{The notation $\widetilde{\Omega}$ hides asymptotically smaller terms.} 
Their theorem, in particular, implies that the sum-of-squares (SoS) hierarchy---a family of increasingly tight convex relaxations---gives an efficient $2$-estimation algorithm for MWIS in claw-free graphs. An obvious open question is whether the above theorem can be generalized to $k$-claw-free graphs.

\subsection{Our contributions }

Our  results are stated via our new notion of  \textit{conditional $\chi$-boundedness}. We say that $G$ is $(t,\gamma)$-conditionally $\chi$-bounded if $\chi(G) \leq \gamma \omega(G)$ whenever ${\alpha(G) \geq t}$. Moreover, a graph class is $(t,\gamma)$-conditionally $\chi$-bounded if every graph in that class is. 
The following theorem (which is a simple combination of known facts) connects SoS to an extremal question of Chudnovsky and Seymour. 

\begin{theorem}
\label{thm: connection} 
Consider a graph class that is closed under clique replacement.\footnote{A clique replacement operation on graph $G$ replaces any vertex $v$ with a clique $K_v$ of arbitrary size and connects each vertex in a clique to every neighbor of $v$. It is easy to see that $k$-claw-free graphs are closed under clique replacement.}  
If the graph class is $(t,\gamma)$-conditionally $\chi$-bounded, then $t$ rounds of SoS gives a factor $\gamma$-estimation for MWIS. 
\end{theorem}

In particular, if $t=O(1)$, then SoS gives a $\gamma$-estimation algorithm in polynomial time. 
Our first main contribution is to rule out this approach, i.e., refuting the possibility to generalize Theorem~\ref{thm:CS}. We show the following result (please refer Theorem~\ref{thm:gen cntegcsgen} for a precise statement).

\begin{theorem}[Simplified] 
\label{thm:cntegcsgen}
For every $k \ge 4$, there exists an infinite family of graphs $\{G_n\}$ such that $G_n$ is a  connected $k$-claw-free graph on $n$ vertices with $\alpha(G_n) = \Omega\left(\frac{n}{\log n}\right)$ and $\chi(G) \ge  f(k) \left(\frac{\omega(G)}{\log \omega(G)}\right)^{k/2}$, for some function $f$.
\end{theorem}

This lower bound almost matches the upper bound provided by~\cite{JedYang}. 
We remark that both upper and lower bounds are tight w.r.t. the state-of-the-art bounds on Ramsey number, that is, improving in either direction requires asymptotically improving the bound on the Ramsey number of a graph. 

Given this theorem, the next obvious open question is whether there is a lower bound on the performance of SoS that ``separates'' SoS from the simple greedy {factor-$k$ approximation} algorithm. While we do not manage to settle this question,  we instead show {in the theorem below} that our construction for Theorem~\ref{thm:cntegcsgen} can be used to construct a bad example for {a} Sherali--Adams strengthening of ${\sf QSTAB}(G)= \{x \in [0,1]^{|V(G)|}: \sum_{i \in Q} x_i \leq 1 (\forall ~\mbox{clique } Q)\}$.

\begin{theorem} [Integrality gap of Sherali--Adams on \qstab] \label{thm:igsap}
  Let $k \ge 4$.  
  For $0 < \epsilon \le \nicefrac{1}{3}$,
  there exists an infinite family of graphs $\{G_n\}$ such that $G_n$ is a connected $k$-claw-free graph on $n$ vertices with
  \footnote{The notations $O_k, \Theta_k, \Omega_k$ hide multiplicative functions in $k$.}
  $\alpha(G_n) = \Omega(n^\epsilon)$ and
  the integrality gap of Sherali--Adams hierarchy on $\qstab(G_n)$ is $\Omega_k(n^{\epsilon})$, even  after $\Omega_k(n^{1-2\epsilon})$ rounds.
\end{theorem}

We remark that this theorem can be contrasted with special cases of $k$-claw-free graphs, e.g., in the bounded degree setting~\cite{bansal2015lovasz}, where the (poly-logarithmic rounds of) Sherali--Adams strengthening of QSTAB provides an optimal approximation factor under the unique games conjecture (UGC)~\cite{austrin2011inapproximability}.  

\paragraph{Discussion of previous work on convex programs.}
Let us compare our result of Theorem~\ref{thm:igsap} with the bounds of Chan and Lau~\cite{doi:10.1137/1.9781611973075.122}. In particular, ~\cite{doi:10.1137/1.9781611973075.122} considers \ksp, a special case of maximum independent set in $(k+1)$-claw-free graph, where we are given a $k$-uniform hypergraph $H$ on $n$ vertices and we are asked to find a maximum matching in $H$. In their work,~\cite{doi:10.1137/1.9781611973075.122} show that Sherali--Adams on the standard LP for \ksp has integrality gap of at least $k-2$, even after $\Omega_k(n)$ rounds. Additionally, they show that, for constant $k$, \qstab for \ksp can be captured by a polynomial size LP and has integrality gap of at most $(k+1)/2$. In contrast, our result of Theorem~\ref{thm:igsap} is for more general problem of maximum independent set in $k$-claw-free graphs, and yields integrality gap of $\Omega_k(n^\epsilon)$, which is a function of $n$, for Sherali--Adams on \qstab with rounds $\Omega_k(n^{1-2\epsilon})$, which is a stronger program than that of~\cite{doi:10.1137/1.9781611973075.122}. However, we remark that, the parameters of Theorem~\ref{thm:igsap} can be adjusted to yield an integrality gap of $g(k)$, for any function $g$, for Sherali--Adams on \qstab with $\Omega_k(n)$ number of rounds, which is linear in $n$ (refer Appendix~\ref{sec:ig_linear_sa} for more details). Thus, our results yield a larger integrality gap for a stronger program for maximum independent set in $k$-claw-free graphs compared to that of~\cite{doi:10.1137/1.9781611973075.122}.

\paragraph{Overview of Techniques:} We give a high-level overview for the proof of Theorem~\ref{thm:cntegcsgen}. 
For simplicity, let us focus on the case of $k=5$ and a slightly weaker bound. The first component of our construction is a Ramsey graph. It is known that there exists an $n$-vertex graph $H$ such that $\alpha(H) = 2$ and there is no clique of size $c \sqrt{n \log n}$ for some constant $c$~\cite{kim}. Therefore this graph is $3$-claw-free (since there is no independent set of size $3$) and $\chi(G) \geq n/\alpha(G) =n/2 \geq \widetilde{\Omega}(\omega(G)^{2})$. 
So, this graph has almost all our desired properties except that $\alpha(H)$ is very small while we need large independent sets. We will ``compose'' several copies of $H$ together to obtain a final graph $G$, ensuring that (i) $\alpha(G)$ can be made arbitrarily large, (ii) $\omega(G)$ is roughly the same as $\omega(H)$, and (iii) the graph remains $5$-claw-free.

The key component of our composition step is a special graph operation that we call \emph{bi-conflict composition}. 
In particular, given any two graphs $H_1, H_2$ having $n$ vertices each, the graph $D = bcc(H_1, H_2)$ is obtained as follows. First, construct a graph $D'$ by connecting $H_1$ and $H_2$ by an arbitrary matching  $M\colon |M| = n$. Next, define $D$ where $V(D) = M$ and for each $e, e' \in M$, we have $(e, e') \in E(D)$ if and only if they are not induced matching edges in $D'$. 
There are two simple properties of graph $D$ that are crucial to our analysis.  
\begin{itemize}
    \item (P1) the maximum claw in $D$ is at most $\min \{\alpha(H_1), \alpha(H_2)\}$. 

    \item (P2) the maximum clique in $D$ is $\omega(D) \leq \omega(H_1) + \omega(H_2) + 1$. 
\end{itemize}
The bi-conflict composition will be used as an analytical tool in our construction. Take $q$ copies of Ramsey graphs $H_1, H_2,\ldots, H_q$ on $n$ vertices. For each $i=1,\ldots, q-1$, connect $H_i$ to $H_{i+1}$ by a matching $M_i\colon |M_i| = n$. Call this graph~$G'$. Our final graph $G$ has vertices $V(G) = \bigcup_i V(M_i)$ and the edges are defined such that $(e,e') \in E(G)$ if $e$ and $e'$ form an induced matching. 
Notice that the structure of graph $G$ for vertices that correspond to $M_i$-edges is roughly the same as $bcc(H_i, H_{i+1})$, which allows us to invoke properties (P1) and (P2).  

Notice that $\alpha(G) = \Omega(q)$ (Pick one matching edge from each $M_i$ when $i$ is odd). Property (P1) guarantees that the maximum claw is at most $4$ (therefore $G$ is $5$-claw-free) and property (P2) guarantees that the value of maximum clique is at most $O(\sqrt{n \log n})$. 
By choosing $q$ to be sufficiently large, we would be done. 
We remark that this construction gives a non-trivial lower bound (albeit weaker than Theorem~\ref{thm:cntegcsgen}) for $k\geq 5$. To make the result work for $k=4$, we need to compose the copies of the Ramsey graph more carefully.

\subsection{Conclusion and open problems}

Our main contribution is to initiate the study of convex relaxation aspects of independent set in $k$-claw-free graphs. The main open question is whether there is any convex relaxation approach that gives a reasonable approximation guarantee for MWIS in $k$-claw-free graphs (for $k \geq 4$). 
Conceptually, we made explicit the implication of Chudnovsky and Seymour's theorem (which can be seen as conditional $\chi$-boundedness), that SoS gives a reasonable approximation guarantee for claw-free graphs. 
We refute the possibility of  generalizing such a result to $k$-claw-free graphs for all $k \geq 4$ and present evidence that this graph family might be a bad instance for SoS.

\subsection{Further related work}


A graph class is said to be $\chi$-bounded if $\chi(G)$ can be upper bounded by  $f(\omega(G))$ where $f$ is some function. The concept of $\chi$-boundedness has been studied extensively in graph theory (see, e.g.,the survey~\cite{scott2020survey} and references therein). 
 In algorithms and optimization, we are mostly concerned with $\chi$-boundedness where $f$ is a linear or close to linear function. Roughly, the ratio  $\chi(G)/\omega(G)$ captures the integrality gap of a convex programming relaxation $\qstab$ in estimating the value of maximum independent set of a graph~\cite{chalermsook2016note}.

As for approximating the maximum weight of the independent set of $k$-claw-free graphs, a local search algorithm due to Berman \cite{Berman00} had remained the best known approximation algorithm for the past two decades. Berman's algorithm achieved a factor of $\frac{k}{2} + \epsilon$ in polynomial time. In 2021, Neuwohner \cite{Neuwohner21} broke the barrier of $\frac{k}{2}$ by improving the factor to $\frac{k}{2} - \frac{1}{63,700,992}$. In a very recent work, Thiery and Ward \cite{TW23} improved the factor to $\frac{k}{2} - \delta_k$ for a constant $\delta_k \geq 0.214$. On the other hand, the problem is NP-hard to approximate better than a factor of $\bigOmega{k / \log k}$ \cite{hazan2006complexity}\footnote{In fact, the hardness even holds for a special case of the problem, namely the unweighted $k$-set packing problem.}.

With respect to convex programs, Chan and Lau~\cite{doi:10.1137/1.9781611973075.122} study the power of standard LP and Sherali--Adams on \ksp, which is a special case of maximum independent set in $(k+1)$-claw-free graph. In \ksp, we are given a $k$-uniform hypergraph $H$ on $n$ vertices and we are asked to find a maximum matching in $H$. In their work,~\cite{doi:10.1137/1.9781611973075.122} show that Sherali--Adams on the standard LP for \ksp has integrality gap of at least $k-2$, even after $\Omega(n/k^3)$ rounds. Additionally, they show that, for constant $k$, \qstab for \ksp can be captured by a polynomial size LP and has integrality gap of at most $(k+1)/2$.

\paragraph{Organization.} Section~\ref{sec:prelim} explains the basic graph-theoretic terminologies. Our graph theoretic result is proved in Section~\ref{sec:graphs}.
All convex relaxation results, as well as the connection with the notion of conditional $\chi$-boundedness, are proved in Section~\ref{sec:convex}.

\section{Preliminaries}
\label{sec:prelim} 
We follow standard graph theoretic notation. 
Given a graph $G=(V(G),E(G))$,  $M \subseteq E(G)$ is a \textit{matching} if no pair of edges in $M$ share a vertex.
Further, a matching $M \subseteq E(G)$ is said to be an \textit{induced matching} if it is an induced subgraph of $G$.
Finally, for a matching $M$, 
$e_i \ne e_j \in M$ is said to be an \textit{intersecting matching pair} if $e_i$ and $e_j$ do not form an induced matching; and $M$ is an \emph{intersecting matching} if every pair in $M$ is an intersecting matching pair.

For $k \ge 3$, a \textit{$k$-claw} is the graph $K_{1,k}$.
For a $k$-claw $T$, the vertex with degree $k$ is called the \textit{central vertex}
of $T$, and the remaining vertices of degree one are called \textit{leaves} of $T$.
 A graph $G$ is said to be \textit{$k$-claw-free} if there exists no $k$-claw as an induced subgraph. 
 For a graph $G$, $\alpha(G)$ is the size of maximum independent set in $G$,  $\omega(G)$ is the size of maximum clique in $G$, and $\chi(G)$ is the  chromatic number of $G$. For weighted case, $\alpha(G)$ and $\omega(G)$ represent the maximum weight independent set and maximum weight clique in $G$ respectively.

The notations $O_k, \Theta_k, \Omega_k$ hide multiplicative functions in $k$.

In the \textit{maximum-weight independent set} problem (MWIS), we are given a graph $G$ together with weights $\{w_v\}_{v \in V(G)}$ on the vertices, and our goal is to find an independent set $S \subseteq V(G)$ with maximum total weights.

\section{Graph theoretic result}
\label{sec:graphs} 

In this section we prove the following theorem.
\begin{theorem} \label{thm:gen cntegcsgen}
   For $k \ge 4$, there exists $n_0$ depending on $k$ such that for infinitely many $n \ge n_0$, there exists a connected $k$-claw-free graph $G$ on $n$ vertices with $\alpha(G) = \Omega\left(\frac{n}{\log n}\right)$ and $\chi(G) \ge f(k) \left(\frac{\omega(G)}{\log \omega(G)}\right)^{k/2}$, for some $f(k)$.
\end{theorem}
To this end, we will use  known results about Ramsey graphs, which we introduce in Section~\ref{sec:ramsey}. 
We explain our construction and analysis in Sections~\ref{sec: construction}. 

\subsection{Ramsey graphs}
\label{sec:ramsey} 

Let $R(s,t)$, for $s \ge 3$ denote the Ramsey number, i.e., $R(s,t)$ is the minimum number such that any graph on $R(s,t)$ vertices has either an independent set of size $s$ or a clique of size $t$. 
\begin{theorem}[~\cite{kim,BOHMAN20091653,bohman2010early}] \label{thm:ramsery lower bound}
    For any $s \ge 3$ and for sufficiently large $t$, $R(s,t) \ge c'_s \cdot t^{\frac{s+1}{2}} (\log t)^{\frac{1}{s-2} - \frac{s+1}{2}}$, for some constant $c'_s$ depending only on $s$.
\end{theorem}

Thus, the above theorem implies that for $s \ge 3$ and sufficiently large $t$, there is a graph on $\lceil R(s,t)-1 \rceil$ vertices that has neither an independent set of size~$s$ nor a clique of size~$t$. We call such a graph an \emph{$(s,t)$-Ramsey graph}.

\begin{corollary} \label{cor:ramsey graphs}
    For $s \ge 3$ and sufficiently large $t$, there is an $(s,t)$-Ramsey graph    on $ c_s \cdot t^{\frac{s+1}{2}} (\log t)^{\frac{1}{s-2} - \frac{s+1}{2}}$ vertices, for some positive constant $c_s$ depending only on $s$.
\end{corollary}

\subsection{Graph construction}
\label{sec: construction} 

\begin{lemma}\label{lm:prop G}
Let $k \ge 4, p \ge 1, \tau \ge 3$, and let $H=(V(H), E(H))$ be a $(k-1,\tau)$-Ramsey graph.
Then, there exists a connected $k$-claw-free graph $G$ on $\Theta(p|V(H)|)$ vertices such that $p \le \alpha(G) \le 3 pk$ and $\omega(G) \le 3\tau$.
\end{lemma}
\begin{proof}
 We construct the graph $G$ in two steps. In the first step, we construct an auxiliary graph $G'=(V(G'), E(G'))$ using  $(k-1,\tau)$-Ramsey graph $H$. Finally, in the second step, using this auxiliary graph $G'$, we construct our  graph $G$.   

\textit{Construction of auxiliary graph $G'$.}
Let $Q = K_{\tau-1}$ be the complete graph on $\tau-1$ vertices. 
We construct the graph $G'$ in two steps: In the first step, we use $H$ and $Q$ to create a  graph $B$, which we call a \emph{block}.
In the second step, we construct the graph $G'$ using $p$ copies of $B$. For $i \in [p]$, we construct block $B_i$ using $Q$ and two copies of $H$ as follows (see Figure~\ref{fig:block}). Let $H^1_i$ and $H^2_i$ be copies of $H$, and let $Q_i$ be a copy of $Q$. We connect each vertex in $H^1_i$ with its respective copy in $H^2_i$ by an edge. Then, we pick an arbitrary set of $\tau-1$ vertices from $H^2_i$ and add an (arbitrary) matching between this set and the vertices of $Q_i$. 
Finally, we connect each $B_i$ to $B_{(i+1) \mod p}$ as follows (see Figure~\ref{fig:twoblock}). We pick an arbitrary set of $\tau-1$ vertices from $H^1_{(i+1) \mod p}$ and add a matching between this set and the vertices of $Q_i$. 
This completes the construction of $G'$. Notice that $G'$ forms a ring structure consisting of  $B_i$s.
We call the new edges that we added in our construction \emph{matching edges}.

\begin{figure}
    \centering
    \includegraphics[scale=0.3]{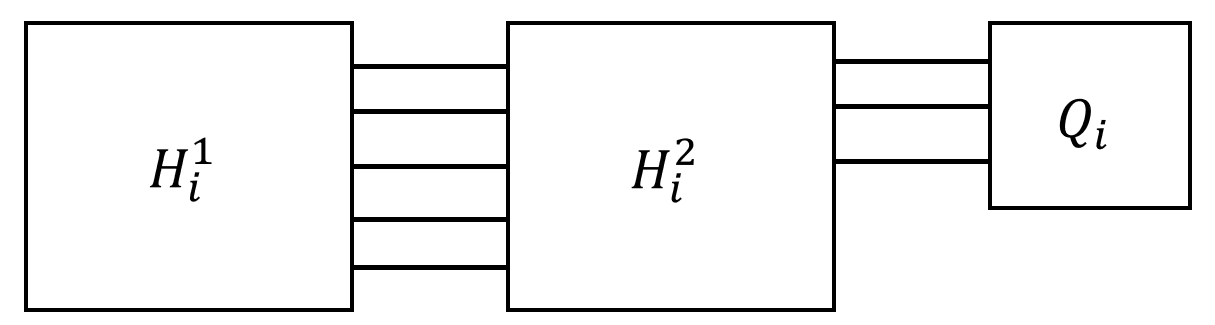}
    \caption{One block $B_i$ of the graph $G'$}
    \label{fig:block}
\end{figure}
 \vspace{-4em}
\begin{figure}
    \centering
    \includegraphics[scale=0.5]{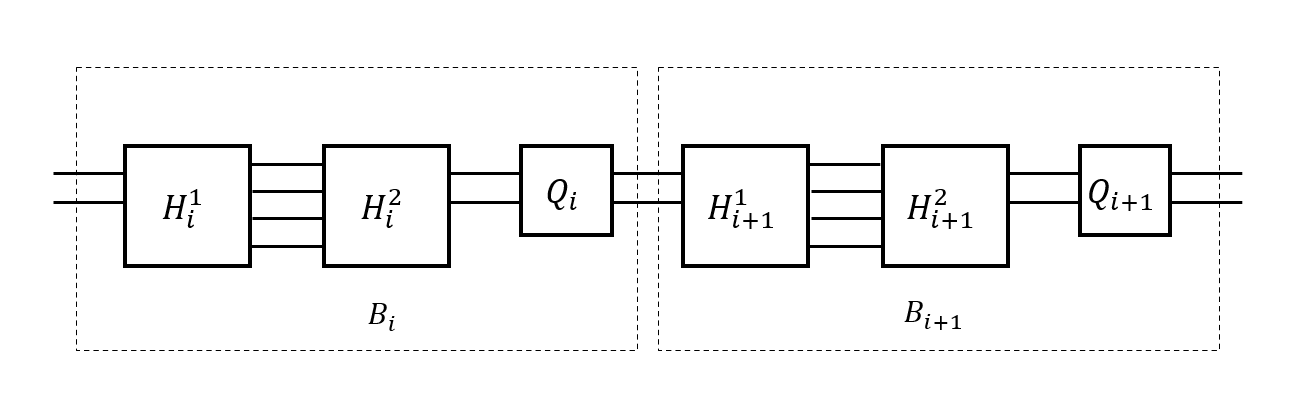}
    \caption{Connecting block $B_i$ and $B_{i+1}$ in $G'$}
    \label{fig:twoblock}
\end{figure}

\textit{Construction of $G$.}
For every matching-edge $e_i$ in $G'$, we create a vertex $v_i$ in $G$. Then, $(v_i,v_j) \in E(G)$ if $e_i$ and $e_j$ is an intersecting matching pair (i.e., they do not form an induced matching).


\textit{Analysis.} 
    Let $n= |V(G)|$, then note $n =  p(q_k(\tau)+2 (\tau-1))$. Hence, $3pq_k(\tau) > n \ge pq_k(\tau)$. Since the Ramsey graph $H$ and $Q=K_{\tau-1}$ used in construction of $B$ are connected, we have that $G'$ is connected. This implies that $G$ is connected. Next we bound $\alpha(G)$. First note that $\alpha(G) \ge p \alpha(H)$ since from each block $B_i$ we can pick matching edges between $H^1_i$ and $H^2_i$ that have endpoints on the vertices of $\alpha(H)$. The vertices in $G$ corresponding to these matching edges form an independent set. For the upper bound, we will show that $\alpha(G) < 3 p \alpha(H)$. Suppose for contradiction $\alpha(G) \ge 3 p\alpha(H)$. 
    Let $I$ be an independent set in $G$ with $|I| \ge 3p \alpha(H)$.
    Let $M$ be the matching-edges of $G'$ corresponding to the vertices in $I$. 
    Let $M_i \subseteq M$ be the edges of $M$ with both endpoints in $B_i$.
    Since $B_i$ and $B_{(i+1) \mod p}$ are connected by $Q_i$ (which is a complete graph), the number of edges of $M$ that are between $B_i$s is at most $p$. Hence, the number of edges of $M$ that lie completely within some $B_i$ is at least $M-p$.
    Since $|M| \ge 3 p \alpha(H)$, it must be that there is a block $B_i$ such that  $|M_i| \ge (3p \alpha(H)-p)/p \ge 2 \alpha(H)$. Since the edges of $M_i$ should have an endpoint in $H^2_i$, it implies that $\alpha(H^2_i) \ge |M_i| \ge 2 \alpha(H)$, which is a contradiction since $H^2_i$ is a copy of $H$.
    Thus, we have that $p\alpha(H) \le \alpha(G) < 3p \alpha(H)$.

    For bounding $\omega(G)$, let $\cC$ be a clique of maximum size in $G$. We claim that $|\cC| \le 3\tau$, which implies that $\omega(G) \le 3\tau$.
    Let $E_\cC$ be the matching-edges of $G'$ corresponding to $\cC$. 
    Let $E'_{\cC} \subseteq E_\cC$ be the edges which have one endpoint in some copy of $Q$, and let $E''_\cC = E_\cC \setminus E'_\cC$. Suppose $E'_\cC \neq \emptyset$, then consider an edge $e' \in E'_\cC$ and suppose $e'$ is incident on $Q_i$ some for $i \in [p]$. Then, observe that every edge in $E'_\cC$ must also be incident on $Q_i$ since  $E'_\cC$ is an intersecting matching and edges incident on $Q_i$ can not intersect with edges incident on $Q_j, j \neq i \in [p]$ by construction. Thus, $|E'_\cC| \le 2 |Q| < 2\tau$. Hence, when $E''_\cC = \emptyset$, we have that $|E_\cC| < 2\tau$ implying $|\cC| < 2\tau$, as desired. For the other case when $E''_\cC \neq \emptyset$  then consider $e'' \in E''_\cC$. We will bound the number of neighbors in $G$ of $e''$ corresponding to the edges of $E'_\cC$ and $E''_\cC$. First, note that $e''$ has at most $|Q| < \tau$ neighbors in $E'_\cC$ since $e''$ can intersect with at most $|Q|=\tau-1$ many edges whose one end point is in $Q_i$. Next, we will show that $e''$ has at most $2(\tau-1)$ neighbors in $E''_\cC$, which implies that $|E''_\cC| < 2\tau$, and hence $|\cC| = |E_\cC| = |E'_\cC| + |E''_\cC| < 3\tau $, as desired.
    To see this, note that $e''$ is incident on two adjacent copies of $H$, and hence the number of (matching) edges between these copies of $H$ that form an intersecting matching together with $e''$ is at $2 \omega(H) = 2(\tau-1)$.

    Finally, we show that $G$ is $k$-claw-free. Suppose for contradiction, there is a $k$-claw $T$ in $G$ with central vertex $v \in T$ and leaves $v_1,\cdots, v_k \in T$. Let $M_T=\{e,e_1,\cdots, e_k\}$ be the matching-edges in $G'$ corresponding to $T$ with $e$ corresponding to $v$ and $e_i$ corresponding to $v_i$, for $i \in [k]$.  Also, let $L_T = \{e_1,\cdots,e_k\}$.
    First consider the case when $e$ is incident between a copy $H'$ of $H$ and a copy $Q'$ of $Q$.
    Let $I \subseteq V(H')$ be the endpoints of edges of $L_T$ in $H'$. We claim that $|I| \ge k-1$. To see this, note that for every edge $e_j \in L_T$ that is not incident on $Q'$ must have one endpoint in $H'$ since $e_j$ and $e$ form an intersecting matching pair. On the other hand, since  $Q'$ is a clique, there can be at most one edge in $L_T$ with endpoint in $Q'$ implying $|I| \ge k-1$.
    As edges in $L_T$ form an induced matching in $G'$, we have that $I$ is an independent set in $H'$. But then $|I| \ge k-1$ which is a contradiction since $\alpha(H') < k-1$. 
    Now, consider the other case when $e$ is between two copies $H'$ and $H''$ of $H$. 
    Let $L'_T \subseteq L_T$ be the edges that are between $H'$ and $H''$, and let $L''_T = L_T \setminus L'_T$. Then, note that $|L'_T| < |L_T|=k$, since otherwise the endpoints of $L'_T$ in $H'$ (or $H''$)  form an independent set of size $k$ in $H'$ (or $H''$ resp.), leading to the contradiction to the fact that $\alpha(H) < k-1$.
    Suppose $(k-2) \le |L'_T| \le (k-1)$, and let $e_j \in L''_T$. Then, since $e_j$ and $e$ form an intersecting matching pair, assume, without loss of generality, $e_j$ has an endpoint in $H'$. But then, since every edge of $L'_T \cup \{e_j\}$ has one endpoint in $H'$, the endpoints of $L'_T \cup \{e_j\}$ in $H'$ form an independent set of size $k-1$ in $H'$ , leading to a contradiction. 
    Finally, if $|L'_T| < k-2$, then there must be at least two edges $e_i,e_j \in L''_T$ incident on one of the two copies of $Q$ adjacent to $H'$ and $H''$. But since $Q$ is a clique, this means $(v_i,v_j) \in E(G)$ contradicting the fact that $T$ is a $k$-claw. \qed
\end{proof}

\subsection{Proof of Theorem~\ref{thm:gen cntegcsgen}}
For $k\ge 4$, let $H$ be a $(k-1,t)$-Ramsey graph obtained from Corollary~\ref{cor:ramsey graphs},  for every sufficiently large $t$. Let $q_k(t) = |V(H)| = c_k t^{\frac{k}{2}} (\log t)^{\frac{1}{k-3} - \frac{k}{2}}$. Then, using $p=2^{q_k(t)}$ and $\tau = t$ along with graph $H$, Lemma~\ref{lm:prop G} produces a graph $G$ on $n := \Theta(2^{q_k(t)}q_k(t))$ vertices such that $\alpha \ge 2^{q_k(t)} = \Omega(n /\log n)$, and $\alpha(G) \le  3k 2^{q_k(t)} = O(nk/\log n)$, and $\omega(G) \le 3t$. Hence, we have
\begin{align*}
    \chi(G) &\ge \frac{n}{\alpha(G)} =\Omega\left(\frac{q_k(t)}{k} \right) =\Omega\left(\frac{c_k}{3k}  \left(\frac{t}{\log t}\right)^{k/2} \right)
    = f(k)\left(\frac{\omega(G)}{\log \omega(G)}\right)^{k/2}, 
\end{align*}
\text{ for some $f(k)$.}

\section{Convex relaxation results}
\label{sec:convex}

\subsection{Convex relaxation prelims}
We explain only necessary terminologies to prove our results. For a complete exposition on sum-of-squares and related convex relaxations, we refer the readers to  excellent survey papers~\cite{laurent2003comparison,fleming2019semialgebraic,rothvoss2013lasserre}. 
Let $K_G$  be the polytope $\{x \in [0,1]^{V(G)}: x_i + x_j \leq 1, \forall (i,j) \in E(G)\}$. The standard LP relaxation for MWIS $\max \{\sum_v w_v x_v: x \in K_G\}$ is known to have integrality gap of at least $\Omega(n)$ on $n$-vertex graphs\footnote{Consider the clique $K_n$ on $n$ vertices and LP assignment $x_i=1/2$ for vertex $i \in K_n$.}. 

\paragraph{Sum-of-squares.} 
The sum-of-square hierarchies (or  Lasserre hierarchies)~\cite{lasserre2001explicit,lasserre2001global,parrilo2000structured,parrilo2003semidefinite} can be applied to (increasingly) tighten any linear program (captured by the level in the hierarchy). For any $t \geq 1$, the $t$-th level of SoS can be computed in time $n^{O(t)}$.

Now we formally define SoS, following the treatment of Rothvoss~\cite{rothvoss2013lasserre}.  

\begin{definition}
Define the $t$-th level of SoS hierarchy ${\sf SoS}_t(K_G)$ as the set of vectors $z \in {\mathbb R}^{2^{V(G)}}$ that satisfy: 
\begin{align*}
M_t(z) &:= (z_{I \cup J})_{|I|, |J| \leq t} \succeq 0 \qquad \text{and}   \\
M^{ij}_t(z) &:= (z_{I \cup J} - z_{I \cup J \cup \{i\}} - z_{I \cup J \cup \{j\}})_{|I|,|J| \leq t} \succeq 0 \quad (\forall (i,j) \in E(G))
\end{align*}
Let ${\sf SoS}^{proj}_t(K_G) = \{(z_i)_{i \in V(G)}: z \in {\sf SoS}_t(K_G)\}$ be the projection on the original variables. 
\end{definition} 
It is standard to view $1$ as $z_{\emptyset}=1$ (so that we have variables $z_I$ for all subsets).

\begin{proposition}[Lemma 8 in~\cite{rothvoss2013lasserre}]
If any solution $x \in K_G$ contains at most $t$ ones, then any $z \in {\sf SoS}^{proj}_t(K_G)$  is a convex combination of integer solutions in $K_G \cap \{0,1\}^{V(G)}$. 
\end{proposition}

\begin{corollary}
\label{cor:exact SoS for small alpha} 
Let $z^* = \arg \max \{\sum_{i \in V(G)} w_i z_i : z \in {\sf SoS}_t(K_G)\}$. 
If $\alpha(G) \leq t$, then the objective value of $z^*$ is exactly the value of maximum weight independent set in $G$.     
\end{corollary}

Our next proposition states that any feasible solution of SoS (at level at least two) can be projected into a feasible solution for QSTAB. 
The proof is somewhat of a folklore nature. Since it has never been written anywhere in the form we need, we provide a proof in Appendix~\ref{sec:prop2 proof} for completeness. 

\begin{proposition}
\label{prop: SoS and QSTAB} 
Let $z \in {\sf SoS}^{proj}_t(G)$ for $t \geq 2$. Then $z \in {\sf QSTAB}(G)$. 
\end{proposition}

\paragraph{Sherali--Adams.} 
Another standard way to increasingly tighten a convex relaxation (such as $K_G$) is via Sherali--Adams hierarchies~\cite{sherali1990hierarchy}. 
Let $G=(V,E)$ be a graph with $V=[n]$ of (unweighted) maximum independent set problem. The \qstab\ LP for $G$ is as follows.
	\begin{eqnarray*}
		\mbox{(\qstab($G$))} & \max & \sum_{i \in [n]} x_i \\ 
		& \mbox{s.t.} & 1- \sum_{ i \in Q} x_i \geq 0\quad \forall \mbox{ clique $Q$}  \\
		& & x_i \geq 0 \quad  \forall i \in [n]
	\end{eqnarray*}

For $\ell \ge 1$, the Sherali Adams hierarchy applied on $\qstab(G)$ is as follows.
	\begin{eqnarray}
		\mbox{(\saplg)} & \max & \sum_{i \in [n]} y_{\{i\}} \nonumber \\ 
		& \mbox{s.t.}   &\forall  S,T \subseteq[n], S \cap T = \emptyset, |S \cup T| \le \ell \text{ following holds.}\nonumber \\
  & \sum_{T' \subseteq T} (-1)^{|T'|} y_{S \cup T'} &-  \sum_{ i \in Q} \sum_{T' \subseteq T} (-1)^{|T'|}  y_{S \cup T' \cup \{i\}}  \geq 0\  \forall \mbox{ clique $Q$} \label{eqn:sa:cliquejunta} \\
  &  \sum_{T' \subseteq T} (-1)^{|T'|} y_{S \cup T'} &\ge 0 \label{eqn:sajunta} \\
    &  \sum_{T' \subseteq T} (-1)^{|T'|} y_{S \cup T' \cup \{i\}} &\ge 0 \quad  \forall i \in [n] \label{eqn:sa:juntavar}\\
    &  y_{\emptyset} &= 1
	\end{eqnarray}
It can be shown that ${\sf SoS}_{\ell+2}$ is at least as strong as \saplg. The formal statement is encapsulated in the following proposition. 

\begin{proposition}
Let $z \in {\sf SoS}_{\ell+2}(K_G)$. Then the solution $\{z_{I}\}_{|I| \leq \ell}$ is feasible for \saplg.     
\end{proposition}

\subsection{Conditional $\chi$-boundedness and SoS}

\newcommand{\gset}{{\mathcal G}}

In this section, we prove Theorem~\ref{thm: connection}. 
Let $H$  be a graph. A clique replacement on graph $H$ replaces a vertex $v \in V(H)$ by a clique $K_v$ of arbitrary size and connects any vertex $u \in K_v$ to all neighbors of $v$. It is an easy exercise to check that $k$-claw-free graphs are closed under clique replacements.

Now we proceed to prove Theorem~\ref{thm: connection}. 
Let $\gset$ be a graph class that is closed under {clique replacement}. Consider an instance $G \in \gset$ and an optimal solution $z^* = \arg \max \{\sum_{i \in V(G)} w_i z_i: z \in {\sf SoS}_t(K_G)\}$. 
If $\alpha(G) \leq t$, we would be done, due to Corollary~\ref{cor:exact SoS for small alpha}.
Otherwise, we consider the projection $z$ of $z^*$ on $V(G)$, so $z \in {\sf QSTAB}(G)$ (due to Proposition~\ref{prop: SoS and QSTAB}). 
By Theorem 1 in~\cite{chalermsook2016note}, there is an independent set in $G$ whose weight is at least $\frac{1}{\gamma} \cdot (\sum_{i \in V(G)} w_i z_i)$, which implies that the integrality gap of this convex relaxation is at most $\gamma$.

\subsection{Integrality gap of Sherali--Adams on \qstab}

In this section, we will show large integrality gap even for the unweighted version of the problem. 
We first show the following theorem.
\begin{theorem} \label{thm:intgapsapl}
    For any graph $G$ on $n$ vertices and $\ell \ge 1$, the integrality gap of \saplg \ is at least $\frac{n}{\alpha(G)(\omega(G) + \ell)}$.
\end{theorem}

\begin{proof}
       To this end, we show the following lemma. 
\begin{lemma}  \label{lm:feasiblesa solution}
    Consider a solution $\hat{y}$ defined as follows. For $A \subseteq [n]$, define
    \[
    \yh_A = 
    \begin{cases}
        1 & \text{ if } A = \emptyset\\
        \frac{1}{\omega(G)+ \ell} & \text{ if } |A| = 1\\
        0 & \text{ otherwise } 
    \end{cases}
    \]
    Then, $\hat{y}$ is a feasible solution to \saplg.
\end{lemma}
\begin{proof}
For $S, T \subseteq [n], |S \cup T| \le \ell, S \cap T = \emptyset$, let $J_{S,T}(y) = \sum_{T' \subseteq T} (-1)^{|T'|} y_{S \cup T'}$.
We will first show that $\hat{y}$ satisfies constraint~(\ref{eqn:sa:cliquejunta} )of \saplg.
Fix some clique $Q$, then we will show that the left hand side of constraint~(\ref{eqn:sa:cliquejunta}): $\sum_{T' \subseteq T} (-1)^{|T'|} y_{S \cup T'} -  \sum_{ i \in Q} \sum_{T' \subseteq T} (-1)^{|T'|}  y_{S \cup T' \cup \{i\}}$ is at least $0$.
For $S$, consider the two cases:  when $S \ne \emptyset$ and  when $S =\emptyset$. For the first case, we have $|S| \ge 1$ and hence $J_{S,T}(\yh) = \yh_S$. Since $\yh_S = 0$ for $|S| \ge 2$, this means $J_{S,T}(\yh) = 0$ for $|S| \ge 2$, as required. For $S= \{a\}, a \in [n]$, this means $J_{S,T}(\yh) = \yh_a$.
Hence, we have
\begin{align*}
 \sum_{T' \subseteq T} (-1)^{|T'|} y_{S \cup T'} -  \sum_{ i \in Q} \sum_{T' \subseteq T} (-1)^{|T'|}  y_{S \cup T' \cup \{i\}}    &= \yh_{\{a\}} - \sum_{i \in Q} \yh_{\{a\} \cup \{i\}}    
\end{align*}
Now if $a \in Q$, this term is $\yh_{\{a\}} -\yh_{\{a\}} =0$,
and if $a \notin Q$, this term is $ \yh_{\{a\}} \ge 0$, due to our construction of $\yh$.

For the second case when $S =\emptyset$, we have that $J_{S,T}(\yh) = \yh_{\{\emptyset\}} - \sum_{j \in T} \yh_{\{j\}}$. Hence, 
\begin{align*}
     \sum_{T' \subseteq T} (-1)^{|T'|} &y_{S \cup T'} -  \sum_{ i \in Q} \sum_{T' \subseteq T} (-1)^{|T'|}  y_{S \cup T' \cup \{i\}}      \\
     &=\yh_{\{\emptyset\}}  - \sum_{j \in T} \yh_{\{j\}}-   \sum_{i \in Q} \yh_{\{i\}} +  \sum_{i \in Q} \sum_{j \in T} \yh_{\{i\} \cup \{j\}} \\
     &=\yh_{\{\emptyset\}}  - \sum_{j \in T} \yh_{\{j\}}-   \sum_{i \in Q} \yh_{\{i\}} +  \sum_{i \in Q \setminus T} \sum_{j \in T} \yh_{\{i,j\}} + \sum_{i \in Q \cap T} \sum_{j \in T} \yh_{\{i\} \cup \{j\}} \\
     &=\yh_{\{\emptyset\}}  - (\sum_{j \in T} \yh_{\{j\}} +   \sum_{i \in Q} \yh_{\{i\}} -  \sum_{i \in Q \cap T}  \yh_{\{i\}} ) \\
     &=\yh_{\{\emptyset\}}  -\sum_{i \in Q \cup T}  \yh_{\{i\}}  
     = 1 - |Q \cup T| \frac{1}{\omega(G) + \ell}   \ge 0,
\end{align*}
since $|Q \cup T| \le \omega(G) + \ell$.

Next consider constraint~(\ref{eqn:sajunta}). From the above observation, we have
\[
J_{S,T} (\yh) =
\begin{cases}
    0 & \text{ if } |S| \ge 2\\
    \yh_a  & \text{ if } S= \{a\}\\
    1 - \sum_{t \in T}\yh_t & \text{ if } S = \emptyset
\end{cases}
\]
Noting the fact that $|T| \le \ell$ and $\yh_t = \nicefrac{1}{(\omega(G) + \ell)}$, we have that $J_{S,T}(\yh) \ge 0$.
Finally consider constraint~(\ref{eqn:sa:juntavar}), and let $J_{S,T,i}(y) = \sum_{T' \subseteq T} (-1)^{|T'|} y_{S \cup T' \cup \{i\}}$, for $i \in [n]$.
Now note that $J_{S,T,i}(\yh) =0$ for $|S| \ge 2$ as before. Hence, first consider the case when $|S|=1$. When $S =\{i\}$ then $J_{S,T,i}(\yh) = \yh_{\{i\}} \ge 0$, otherwise for $S=\{j\}, j \ne i$, we have $J_{S,T,i}(\yh) = 0$. Finally, when $S=\emptyset$ then $J_{S,T,i}(\yh) = \sum_{T' \subseteq T}(-1)^{|T'|} \yh_{T' \cup \{i\}} = \yh_{i}$ if $i \notin T$ otherwise $J_{S,T,i}(\yh) = \yh_{i}- \yh_{i} =0$, if $i \in T$.
\end{proof} 
\qed
\end{proof}
\subsubsection{Proof of Theorem~\ref{thm:igsap}.}
For given constant $k \ge 4$, $0 < \epsilon \le \nicefrac{1}{3}$, and sufficiently large $n$, we will show a connected $k$-claw-free graph $G_n$ on $\Theta(n)$ vertices such that $\alpha(G_n) =\Omega(n^\epsilon)$ and the integrality gap of \saplgn, for $\ell=\Theta_k(n^{1- 2\epsilon})$, is at least $\Omega_k(n^{\epsilon})$. To this end, let $t$ be such that Corollary~\ref{cor:ramsey graphs} yields a $(k-1,t)$-Ramsey graph $H_n$ on $\Theta(n^{1-\epsilon})$ vertices. 
Note that $t = O_k( n^{\frac{1-\eps}{k/2}} \log n ) = O_k(n^{\frac{1-\eps}{2}} \log n)$, since $k \ge 4$.
Let $G_n$ be the graph obtained from Lemma~\ref{lm:prop G} with given value of $k$, $p=\Theta(n^{\epsilon}), \tau=t$, and $H_n$. Note that $G_n$ has  $\Theta(n)$ vertices.
Then, we have that $\Omega(n^\epsilon) \le \alpha(G_n) \le O_k(n^\epsilon)$ and $\omega(G_n) =  O_k(n^{\nicefrac{(1-\epsilon)}{2}} \log n)$.
Now using $\ell = \Theta_k(n^{1- 2\epsilon})$ in Theorem~\ref{thm:intgapsapl}, the integrality gap of \saplgn $\ge \frac{n}{\alpha(G) (\omega(G) + \ell)}    = \Omega_k(n^{\epsilon})$, since $\omega(G_n) = O(\ell)$.
\qed

%
%
%
\bibliographystyle{splncs04}
\bibliography{papers}

\appendix

\section{A bad example for \qstab}
\label{sec: qstab bad} 

Let $k \geq 3$. 
We show that the \qstab constraints alone are not sufficient for giving a good approximation for MWIS in $k$-claw-free graphs. 
From Theorem~\ref{thm:ramsery lower bound}, let $G$ be an $n$-vertex graph that has no independent set of size $k$ and no clique of size $t$ where $n = c_k (\frac{t}{\log t})^{(k+1)/2}$ where $c_k$ is a constant depending on $k$. 
This gives $t = \Theta((n \log n)^{\frac{2}{k+1}})$. 
Since $\alpha(G) <k$, this graph is immediately $k$-claw-free. 

Define the variable $z \in {\mathbb R}^{V(G)}$ where $z_i = 1/t$. Clearly, $z \in \qstab(G)$ and the objective value is $n/t \geq (\frac{n}{\log n})^{1-2/(k+1)}$, while $\alpha(G) =k$. This gives us an arbitrary large ratio between the solution in $\qstab$ and an optimal integer solution.  

\section{Proof of Proposition~\ref{prop: SoS and QSTAB}}
\label{sec:prop2 proof}
It suffices to show that $\sum_{i \in Q} z_i \leq 1$ for all clique $Q$ in $G$. 
Let $\{\textbf{v}_I\}_{I \subseteq V(G)}$ be the vector representations of solution $z$. In particular, we know the following connection between solution $z$ and vectors $v$ (Lemma 7 of~\cite{rothvoss2013lasserre})\begin{equation} \label{eq:vectors and z}
\forall I, J \subseteq V(G): v_I \cdot v_J = z_{I \cup J}    
\end{equation}

We remark that the statements hold: 

\begin{itemize}
    \item $||v_{\emptyset}||^2 =1$.  

    \item $v_{\emptyset} \cdot v_i = ||v_i||^2$. Proof: Using Equation~\ref{eq:vectors and z}, we have that $v_i \cdot v_{\emptyset} = z_{i} = v_i \cdot v_i = ||v_i||^2$. 

    \item $v_i \cdot v_j = 0$ for all $(i,j) \in E(G)$. This is due to Corollary 3 in~\cite{rothvoss2013lasserre}. 
\end{itemize}

Finally, we can write $\sum_i z_i$ as $\sum_i v_i \cdot v_{\emptyset} = \sum_i ||v_i||^2 = (\sum_i v_i) \cdot (\sum_i v_i) = ||\sum_i v_i||^2$. 
Notice that the term $\sum_i v_i \cdot v_{\emptyset} = (\sum_i v_i) \cdot v_{\emptyset} \leq ||\sum_i v_i||$. Combining these, we have that $||\sum_i v_i|| \leq 1$ which implies that $\sum_i z_i \leq 1$. 
This completes the proof that $z \in \qstab(G)$.

\section{Integrality gap of linear number of rounds of Sherali--Adams on \qstab}
\label{sec:ig_linear_sa}
We will show the following. We remark that, in order to obtain cleaner parameters, we have  not tried to optimize them. 
\begin{theorem}\label{thm:ig_linear_sa}
  For constant $k\ge 4$, there is an infinite family of $k$-claw-free graphs $\{G_n\}$ on $\Theta(n)$ vertices with  $\alpha(G_n) \ge k$ such that Sherali--Adams on $\qstab(G_n)$ with $\Omega({n}/{f(k)})$ rounds, has integrality gap at least $\Omega({f(k)}/{k^2})$, for any $f(k) = \Omega(k^3)$.  
\end{theorem}
For comparison, let us recall the bounds of~\cite{doi:10.1137/1.9781611973075.122}.

\begin{theorem}[Theorem 1.3 of~\cite{doi:10.1137/1.9781611973075.122}]\label{thm:cl_ig}
    For constant $k\ge 4$, there are $k$-uniform hypergraphs in which the integrality gap for the Sherali--Adams hierarchy on (LP) is at least $k-2$, even after $\Omega(n/k^3)$ rounds where $n$ denotes the number of vertices.
\end{theorem}

Using $f(k)=\Theta(k^3)$ in our theorem, we obtain the following corollary.
\begin{corollary}
      For constant $k\ge 4$, there is an infinite family of $k$-claw-free graphs $\{G_n\}$ on $\Theta(n)$ vertices with  $\alpha(G_n) \ge k$ such that Sherali--Adams on $\qstab(G_n)$ with $\Omega({n}/{k^3})$ rounds, has integrality gap at least $\Omega(k)$.
\end{corollary}

Thus, with asymptotically same number of rounds, our bad example has asymptotically the same integrality gap but for Sherali--Adams on \qstab, which is a stronger program than Sherali--Adams on LP. Further, Theorem~\ref{thm:ig_linear_sa} is able to obtain arbitrary large integrality gap depending on the function $f(k)$, whereas Theorem~\ref{thm:cl_ig} has a fixed gap. Hence, for fixed $k$, Theorem~\ref{thm:ig_linear_sa} yields an arbitrary large gap for Sherali--Adams on \qstab, even after $\Omega(n)$ rounds.

\begin{proofof}{Theorem~\ref{thm:ig_linear_sa}}
Given $k \ge 4$, we will show infinitely many graphs $G_n$ on $\Theta(nk)$ vertices with the desired properties. Let $n_0$ be such that for all $n \ge n_0$ it holds that
\begin{equation} \label{eqn:1}
    3\left(\frac{n}{k\cdot c_{k-1}}\right)^{1/2} \log n \le \frac{nk}{f(k)}
\end{equation}
where $c_{k-1}$ is the function of Corollary~\ref{cor:ramsey graphs}. Note that there exists such $n_0$ since $k$ is fixed and $n^{1/2} \log n = o(n)$.
Now, let $t$ be such that Corollary~\ref{cor:ramsey graphs} yields a $(k-1,t)$-Ramsey graph $H_n$ on $n \ge n_0$ vertices. 
Note that $t \le  \left(\frac{n}{k \cdot c_{k-1}}\right)^{2/k} \log n \le   \left(\frac{n}{k \cdot c_{k-1}}\right)^{1/2} \log n $ since $k \ge 4$.
Let $G_n$ be the graph obtained from Lemma~\ref{lm:prop G} with the given value of $k$, $p=k, \tau=t$, and $H_n$. Note that $G_n$ has  $\Theta(nk)$ vertices.
Then, we have that $k \le \alpha(G_n) \le O(k^2)$ and $\omega(G_n) \le 3 \left(\frac{n}{k \cdot c'_{k-1}}\right)^{1/2} \log n$.
Now using $\ell = \frac{nk}{f(k)}= \Theta\left(\frac{|V(G_n)|}{f(k)}\right)$ in Theorem~\ref{thm:intgapsapl} and using the fact $\omega(G_n) \le \ell$ due to Equation~\ref{eqn:1}, the integrality gap of \saplgn $\ge \frac{\Theta(nk)}{\alpha(G_n) (\omega(G_n) + \ell)}    = \Omega\left(\frac{f(k)}{k^2}\right)$.
\end{proofof}

\end{document}